\documentclass{amsart}[11pt]
\usepackage{amsmath, amsfonts, amsthm, amssymb,mathtools}
\usepackage{subfigure}
\usepackage{stmaryrd}
\usepackage{verbatim}
\usepackage{hyperref}
\usepackage{color}
\usepackage{xcolor}
\usepackage{ulem}
\usepackage{graphicx}
\usepackage{enumerate}
\usepackage{bbm}
\numberwithin{equation}{section}
\newtheorem{theorem}{Theorem}[section]
\newtheorem{corollary}[theorem]{Corollary}
\newtheorem{lemma}[theorem]{Lemma}

\theoremstyle{definition}

\newtheorem{assumption}[theorem]{Assumption}

\makeatletter
\newcommand*{\rom}[1]{\expandafter\@slowromancap\romannumeral #1@}
\makeatother

\newcommand{\ind}{1\hspace{-2.1mm}{1}} %Indicator Function
\newcommand{\I}{\mathtt{i}}
\newcommand{\eps}{\varepsilon}

\newcommand{\EE}{\mathbb{E}}

\newcommand{\RR}{\mathbb{R}}
\newcommand{\QQ}{\mathbb{Q}}

\newcommand{\PP}{\mathbb{P}}
\newcommand{\RV}{\mathcal{R}}
\newcommand{\Crm}{\mathrm{C}}

\newcommand{\m}{\mathfrak{m}}

\newcommand{\A}{\mathfrak{a}}

\newcommand{\vm}{\mathfrak{v}_-}
\newcommand{\vp}{\mathfrak{v}_+}
\newcommand{\M}{\mathrm{M}}
\newcommand{\D}{\mathrm{d}}

\newcommand{\E}{\mathrm{e}}
\newcommand{\Drm}{\mathrm{D}}

\newcommand{\LDP}{\mathrm{LDP}}

\newcommand{\sgn}{\mathrm{sgn}}

\newcommand{\Vv}{\mathcal{V}}

\newcommand{\Oo}{\mathcal{O}}

\newcommand{\Dd}{\mathcal{D}}

\newcommand{\ccm}{\underline{\mathfrak{c}}}
\newcommand{\ccp}{\overline{\mathfrak{c}}}
\newcommand{\LLm}{\underline{\Lambda}^{*}}
\newcommand{\LLp}{\overline{\Lambda}^{*}}

\newcommand{\rrho}{\overline{\rho}}

\newcommand{\Ff}{\mathcal{F}}

\newcommand{\ggm}{\underline{\gamma}}
\newcommand{\ggp}{\overline{\gamma}}

\newcommand{\xhat}{X^{g}}
\newcommand{\lambhat}{\Lambda^{g}}

\graphicspath{{Figures/}  }

\title{Small-time moderate deviations for the randomised Heston model}
\date{\today}
\author{Antoine Jacquier, Fangwei Shi}
\address{Department of Mathematics, Imperial College London}
\email{a.jacquier@imperial.ac.uk, fangwei.shi12@imperial.ac.uk}

\date{\today}
\begin{document}
\maketitle

\begin{abstract}
We extend previous large deviations results for the randomised Heston model to 
the case of moderate deviations.
The proofs involve the G\"artner-Ellis theorem and sharp large deviations tools. 
\end{abstract}

\section{Introduction}\label{sec:intro}
Classical stochastic volatility models are known to provide an overall good fit of option price data 
(or of the so-called implied volatility surface), except for short maturities;
in this particular region, adding jumps has historically provided a good patch, 
at the expense of complicated hedging, and more recently rough volatility models~\cite{Alos, BayerRough, Euch2, Fukasawa, gatheral14, guennoun14, JMMVix} have shown to out-perform while preserving continuity of the sample paths.
Building on the intuition that these refinements somehow capture a certain kind of uncertainty around the starting time of the process, 
a randomised version of the Heston model~\cite{Heston} was proposed in~\cite{JS16, mechkov15}, 
where the starting point of the variance process is considered random.
The authors showed there that this extra source of randomness generates the desired behaviour 
of implied volatility for small times. 
Mathematically, this was proved showing that the underlying stock price process satisfies 
some large deviations principles with specific rates of convergence.

Moderate deviations, although formally equivalent to large deviations, 
however usually provide more efficient ways (from a numerical point of view) to compute limiting probabilities.
Introduced in~\cite{Rubin}, they have become an increasingly useful tool
in probability and statistical Physics, as can be found in~\cite{Deo, Acosta, dembo96, Liming}.
They have also recently appeared in mathematical finance in order to provide a different, yet somehow more useful view on asymptotics, and important results in this direction can be studied in~\cite{BayerRough2, FrizGerholdPinter, JS18}.

This paper builds upon the large deviations results from~\cite{JS16} and provide their moderate deviations counterparts, in the context of small-time behaviour of the randomised Heston model;
contrary to large deviations, the moderate deviations rate functions are here available in closed form,
hence allowing for more efficient and quicker computations.
In passing, we provide (Theorems~\ref{thm: ldp_ThinTail_rescaled} and~\ref{thm:fat-tail-euro-call})
unusual examples of moderate deviations rate function which does not have a quadratic form.
We gather some technical results and background in the appendix.

\textbf{Notations}
Let~$\RR_+ := [0,\infty)$, $\RR_+^* := (0,\infty)$, and $\RR^* := \RR\setminus\{0\}$. 
For two functions~$f$ and~$g$ we write $f\sim g$ as~$x$ tends to $x_0$
if $\lim\limits_{x\rightarrow x_0} f(x)/g(x) =1$.
Finally, for a sequence $(Y_t)_{t\geq 0}$ satisfying a large deviations principle as~$t$ tends to zero
with speed~$g(t)$ and good rate function~$\Lambda$ we use the notation
$Y\sim\LDP(g(t), \Lambda)$.

%%%%%%%%%%%%%%%%%%%%%%%%%%%%%%%%
%%%%%%%%%%%%%%%%%%%%%%%%%%%%%%%%
\section{Model description}\label{sec:modelsetup}
On a given filtered probability space
$(\Omega, \Ff, (\Ff_t)_{t\geq 0}, \PP)$ supporting two independent Brownian motions~$W^{(1)}$
and~$W^{(2)}$, 
we consider the following dynamics for a log-stock price process $(X_t)_{t\geq 0}$:
\begin{equation}\label{eq:rHeston}
\begin{array}{rll}
\D X_{t} & = \displaystyle
 -\frac{1}{2}V_{t} \D t + \sqrt{V_{t}}\left(\rho\,\D W^{(1)}_t + \sqrt{1-\rho^2}\,\D W^{(2)}_t\right), 
 & X_{0} = 0,\\
\D V_{t} & = \displaystyle \kappa (\theta-V_{t})\D t + \xi \sqrt{V_{t}} \D W^{(1)}_{t}, 
& V_{0} = \Vv,
\end{array}
\end{equation}
with $\kappa, \theta, \xi>0$ and $\rho\in[-1,1]$.
This corresponds to the randomised version of the classical Heston model~\cite{Heston},
as recently proposed and analysed in~\cite{JS16, mechkov15}.
We assume that~$\Vv$ is a continuous random variable independent of the filtration~$(\Ff_t)_{t\geq 0}$,
and that the interior of its support reads 
$(\vm,\vp)\subseteq\RR_+^*$.
Further assume that its moment generating function 
$\M_{\Vv}(u):=\EE(\E^{u\Vv})$
is well defined on an open interval containing the origin, 
and denote $\m:=\sup\{u:\EE(\E^{u\Vv})<\infty\}$.
We shall distinguish three separate behaviours for the randomisation~$\Vv$:
bounded-support ($\vp<\infty$), thin-tail ($\m=\infty$, $\vp=\infty$), and fat-tail ($\m<\infty$, $\vp=\infty$).
Following~\cite{JS16}, we introduce the following assumptions characterising the thin-tail and fat-tail cases:

\begin{assumption}[Thin tails]\label{Assu:V0}
$\vp = \infty$ and~$\Vv$ admits a smooth density~$f$ with
$\log f(v)\sim -l_1 v^{l_2}$ as~$v$ tends to infinity, for some~$(l_1,l_2)\in \RR_+^*\times(1,\infty)$.
\end{assumption}
\begin{assumption}[Fat tails]\label{assu:fat-tail-2}
There exists $(\gamma_0,\gamma_1,\omega) \in \RR^*\times\RR\times\{1,2\}$, 
such that the following asymptotics hold for the cumulant generating function (cgf) of~$\Vv$ as~$u$ tends to~$\m$ from below:
\begin{equation}\label{eq:fat-tail assu2}
\log\M_\Vv(u) = 
\left\{
\begin{array}{ll}
\displaystyle \gamma_0\log(\m-u) + \gamma_1 + o(1), &\text{ for }\omega=1,\gamma_0<0,\\
\displaystyle \frac{\gamma_0}{\m-u} \left(1+ \gamma_1(\m-u)\log(\m-u) + \Oo(\m-u)\right), & \text{ for }\omega = 2,\gamma_0>0,
\end{array}
\right.
\end{equation}
and
\begin{equation}\label{eq:fat-tail assu3}
\frac{\M'_\Vv(u)}{\M_\Vv(u)} =
\left\{
\begin{array}{ll}
\displaystyle \frac{|\gamma_0|}{\m-u}\left(1 + o(1)\right), &\text{ for }\omega =1,\gamma_0<0,\\
\displaystyle \frac{\gamma_0 }{(\m-u)^2}
\left(1- \gamma_1(\m-u) + o\left(\m-u\right)\right), &\text{ for }\omega = 2,\gamma_0>0,
\end{array}
\right.
\end{equation}
\end{assumption}
Common continuous distributions fit into this framework, in particular
the uniform distribution (bounded support), the folded Gaussian distribution, 
the Gamma distribution (Assumption~\ref{assu:fat-tail-2} with $\omega=1$), 
and the noncentral Chi-squared (Assumption~\ref{assu:fat-tail-2} with $\omega=2$).

Before stating the main results of the paper, let us recall some information on the cumulant generating function
of~$X_t$, which will be essential for the rest of the analysis.
As proved in~\cite{albrecher07}, 
the moment generating function of~$X_t$ in the standard Heston model 
(where~$\Vv$ is a Dirac mass at $v_0>0$) 
admits the closed-form representation
$\M(t,u) = \exp\left(\Crm(t, u)+\Drm(t, u)v_0\right)$, for any $u\in\Dd_\M^t\subset\RR$,
where 
\begin{equation}\label{eq:HestonMGF}
\left\{
\begin{array}{rl}
\Crm(t,u) & := \displaystyle \frac{\kappa\theta}{\xi^{2}}
\left[(\kappa- \rho\xi u-d(u))t - 2\log\left(\frac{1-g(u)\E^{-d(u)t}}{1-g(u)}\right)\right],\\
\Drm(t,u) & := \displaystyle \frac{\kappa- \rho\xi u-d(u)}{\xi^{2}}\frac{1-\exp\left(-d(u)t\right)}{1-g(u)\exp\left(-d(u)t\right)},\\
d(u) & : = \displaystyle \left((\kappa - \rho\xi u)^{2}+\xi^{2} u (1-u)\right)^{1/2}
\qquad\text{and}\qquad
g(u) : = \frac{\kappa- \rho\xi u-d(u)}{\kappa-\rho\xi u +d(u)}.
\end{array}
\right.
\end{equation}
\begin{comment}
In the proof of~\cite[Lemma 6.1]{forde12}, the authors showed that the functions~$d$ and~$g$ have the following behaviour
as~$t$ tends to zero:
\begin{equation}\label{eq:heston_asymp}
d\left(\frac{u}{t}\right) = \I \frac{d_0 u}{t} + d_{1} + \Oo (t)
\qquad\text{and}\qquad
g\left(\frac{u}{t}\right) = g_{0} - \I \frac{g_1}{u}t + \Oo (t^2),
\end{equation}
with
$d_{0} := \xi\rrho\,\sgn(u)$, 
$d_{1} := \I \displaystyle \frac{2\kappa\rho-\xi}{2\rrho} \sgn(u)$, 
$g_{0} := \displaystyle \frac{\I\rho-\rrho\sgn(u)}{\I\rho+\rrho\,\sgn(u)}$ and 
$g_{1} := \displaystyle \frac{(2\kappa-\rho\xi)\sgn(u)}{\xi\rrho(\rrho+\I \rho\,\sgn(u))^{2}}$.
\end{comment}
Introduce further the real numbers $u_-\leq 0$ and $u_+\geq 1$ and the function $\Lambda:(u_-, u_+)\to\RR$:
\begin{equation}\label{eq:lambda_H}
\begin{array}{rl}
& \left\{
\begin{array}{rl}
u_{-} &:= \displaystyle \frac{2}{\xi\rrho}\arctan\left(\frac{\rrho}{\rho}\right)\ind_{\{\rho<0\}}
- \frac{\pi}{\xi}\ind_{\{\rho=0\}}
+ \frac{2}{\xi\rrho}\left(\arctan\left(\frac{\rrho}{\rho}\right)-\pi\right)\ind_{\{\rho>0\}},\\
u_{+} &:= \displaystyle \frac{2}{\xi\rrho}\left(\arctan\left(\frac{\rrho}{\rho}\right)+\pi\right)\ind_{\{\rho<0\}}
+\frac{\pi}{\xi}\ind_{\{\rho=0\}}
+ \frac{2}{\xi\rrho}\arctan\left(\frac{\rrho}{\rho}\right)\ind_{\{\rho>0\}},
\end{array}
\right.
\\
& \displaystyle \Lambda(u) := \frac{u}{\xi(\rrho\mathrm{cot}\left(\xi\rrho u/2\right)-\rho)}.
\end{array}
\end{equation}
The pointwise limit of the (rescaled) cumulant generating function of~$X_t$ then reads~\cite{fordeJac}
$$
\lim_{t\downarrow 0}t\log\M\left(t,\frac{u}{t}\right) = \Lambda(u)v_0,
\qquad\text{for any }u\in (u_-,u_+),
$$
and the function~$\Lambda$ is well defined, smooth, strictly convex on $(u_-, u_+)$, and infinite elsewhere.

%%%%%%%%%%%%%%%%%%%%%%%%%%%%%%%%%
%%%%%%%%%%%%%%%%%%%%%%%%%%%%%%%%%
\section{Moderate deviations}\label{sec:mdp}
Moderate deviations classically arise as rescaled large deviations;
in our setting, they take the following form:
for $\alpha \ne 0$, define the process~$X^{(\alpha)}$ pathwise via $X^{(\alpha)}_t := t^{-\alpha}X_t$. 
Moderate deviations for the sequence~$(X_t)_{t\geq 0}$ as~$t$ tends to zero 
are equivalent to large deviations for~$(X_t^{(\alpha)})_{t\geq 0}$ and can,
in our framework, be derived from finite-dimensional tools using the G\"artner-Ellis theorem.
The assumptions on the behaviour of the randomisation~$\Vv$ yield 
different rate functions and speed for the moderate deviations regime, 
which we analyse sequentially below.

%%%%%%%%%%%%%%%%%%%%%%%%%%%%%%%%%
\subsection{Distribution with bounded support}\label{sec:bddsupp}
We first start with the case where the random initial distribution of~$\Vv$ has bounded support, 
in which case the following holds:
\begin{theorem}\label{thm: ldp_bdd_rescaled}
If~$\vp$ is finite then for any~$\gamma \in (0,1)$, 
$X^{(\alpha)}\sim\LDP\left(t^{\gamma}, \frac{x^2}{2\vp}\right)$ holds with $\alpha := \frac{1}{2}(1-\gamma)$.
\end{theorem}
Since $\vp$ is finite, $\m$ is infinite.
One of the striking feature of moderate deviations is that, contrary to classical large deviations,
the rate function is usually available analytically, and often of quadratic form~\cite{FrizGerholdPinter, guillin02, guillin04}.
\begin{proof}
Let $\alpha, \gamma \in (0,1)$.
Notice that
\begin{equation}\label{eq:TowerProperty}
\M(t,u) := \EE \left(\E^{uX_t} \right)
= \EE\left(\EE\left(\E^{uX_t}|\Vv\right)\right)
= \EE\left(\E^{\Crm(t, u) + \Drm(t, u)\Vv}\right)
= \E^{\Crm(t, u)} \M_{\Vv}\left(\Drm(t, u)\right),
\end{equation}
where the functions $\Crm$ and $\Drm$ are the components of the moment generating function 
of the standard Heston model in~\eqref{eq:HestonMGF}.
Then for any $t>0$, the rescaled cumulant generating function of $X_t^{(\alpha)}$ reads
$$
\Lambda^{(\alpha)}_\gamma\left(t,\frac{u}{t^{\gamma}}\right) 
:= t^\gamma\log\EE\left[\exp\left(\frac{uX_t^{(\alpha)}}{t^\gamma}\right)\right]
= t^\gamma\log\EE\left[\exp\left(\frac{uX_t}{t^{\gamma+\alpha}}\right)\right]
= t^\gamma \Crm\left(t,\frac{u}{t^{\gamma+\alpha}}\right) 
+ t^\gamma\log\M_{\Vv}\left(\Drm\left(t,\frac{u}{t^{\gamma+\alpha}}\right)\right),
$$
for all~$u\in\RR$ such that the left-hand side exists.
Lemma~\ref{lemma:CDasymp} implies that $(\gamma+\alpha)$ 
has to be less than one in order to obtain a non-trivial behaviour. 
Let us first prove the following claim: 
for $\vp<\infty$, $\lim\limits_{u\uparrow \infty}u^{-1}\log\M_{\Vv}(u) = \vp$.
If~$F_{\Vv}$ denotes the cumulative distribution function of~$\Vv$, then
$$
\M_{\Vv}(u) = \EE\left(\E^{u\Vv}\right)\leq \exp(u\vp)\int_{[\vm, \vp]} F_{\Vv}(\D v) = \exp(u\vp).
$$
For any small $\eps>0$, fix $\delta \in(0,\eps\vp/2)$, so that
$$
\frac{\log\M_{\Vv}(u)}{u\vp}
\geq \frac{1}{u\vp}\log\left(\int_{\vp-\delta}^{\vp}\E^{uv}F_{\Vv}(\D v)\right)
\geq \frac{1}{u\vp}\log\left(\E^{u(\vp-\delta)}\PP\left(\Vv \geq \vp-\delta\right)\right)
= 1-\frac{\delta}{\vp} + \frac{\log\PP\left(\Vv \geq \vp-\delta\right)}{u\vp},
$$
since~$\vp$ is the upper bound of the support,
therefore $\PP(\Vv\geq \vp - \delta)$ is strictly positive, and the claim follows.
From this claim, as~$t$ tends to zero, we deduce the asymptotic behaviour
\begin{equation*}
\Lambda^{(\alpha)}_\gamma\left(t,\frac{u}{t^\gamma}\right) = 
\left\{\begin{array}{lll}
 \displaystyle \Oo\left(t^\gamma\right) + \vp\Lambda(u)t^{\gamma-1},
 & \text{if }\gamma+\alpha = 1,
 & \text{for all }u\in(u_-,u_+),\\
\displaystyle o\left(t^\gamma\right) + \frac{\vp}{2}u^2t^{1-\gamma-2\alpha},
 & \text{if }\gamma+\alpha<1,
 & \text{for all }u\in\RR.
\end{array}
\right.\end{equation*}
Since $\alpha\ne 0$, the non-degenerate result is obtained if and only if
$1-\gamma-2\alpha=0$, i.e.~$\alpha = \frac{1-\gamma}{2}$, 
and the proof follows from the G\"artner-Ellis theorem~\cite[Theorem~2.3.6]{dembo92}.
\end{proof}
%%%%%%%%%%%%%%%%%%%%%%%%%%%%%%%%%%%%%%%%%%%%%%
%%%%%%%%%%%%%%%%%%%%%%%%%%%%%%%%%%%%%%%%%%%%%%
\subsection{Thin-tail distribution}\label{sec:thin-tail}
With $l_1,l_2$ given in Assumption~\ref{Assu:V0},
we introduce the following two special rates of convergence
$\frac{1}{2}<\ggm<1<\ggp$, 
and two positive constants~$\ccm$,~$\ccp$:
\begin{equation}\label{eq:SpecialRates}
\ggm := \frac{l_2}{1+l_2},
\qquad\qquad
\ggp := \frac{l_2}{l_2-1},
\qquad\qquad
\ccm := (2l_1l_2)^{\frac{1}{1+l_2}},
\qquad\qquad
\ccp := (2l_1l_2)^{\frac{1}{1-l_2}},
\end{equation}
and define the function $\LLm: \RR\to\RR_+$ by
\begin{equation}\label{eq:I}
\LLm(x) := \frac{\ccm}{2\ggm} x^{2\ggm},\qquad\text{for any }x\text{ in }\RR.
\end{equation}
Introduce further
$$
\LLp(x) :=
\sup_{u\in(u_-,u_+)}\left\{ux - \frac{\ccp}{\ggp}
2^{\ggp-1}\Lambda(u)^{\ggp}\right\},
\quad\text{for all } x\in\RR.
$$
with $\Lambda$ and~$u_{\pm}$ in~\eqref{eq:lambda_H}.
The moderate deviations principle then takes the following form:
\begin{theorem}\label{thm: ldp_ThinTail_rescaled}
Under Assumption~\ref{Assu:V0}, the following statements hold as $t$ tends to zero:
\begin{enumerate}[(i)]
\item for any~$\gamma \in (0,\overline{\gamma})$,
$X^{(\alpha)} \sim \LDP(t^\gamma, \LLm)$ with 
$\alpha = \frac{1}{2}(1 - \gamma / \ggm)$;% and $\LLm$;
\item if~$\gamma = \overline{\gamma}$, then $X^{(\alpha)}\sim\LDP(t^{\overline{\gamma}}, \LLp)$ 
with $\alpha = 1-\overline{\gamma}$.
\end{enumerate}
\end{theorem}
Let us first state and prove the following short technical lemma.
Recall~\cite{bingham89} that, for $\A>0$, a function~$f: (\A,\infty)\to\RR_+^*$  is said to be regularly varying with index~$l\in\RR$ (and we write $f\in\RV_l$)
if
$\lim\limits_{x\uparrow \infty}f(\lambda x)/f(x) = \lambda^l$, for any~$\lambda>0$.
When $l=0$, the function is called slowly varying.

\begin{lemma}\label{lemma: log-mgf-asymptotics}
If $|\log f|\in \RV_l$ ($l>1$), then
$
\log\M_{\Vv}(z) \sim \left(l-1\right) \left(\frac{z}{l}\right)^{\frac{l}{l-1}}\psi(z)$
at infinity,
with~$\psi\in\RV_0$ defined as 
$$
\psi(z) := \left(\frac{z}{|\log f|^\leftarrow(z)}\right)^\leftarrow z^{\frac{l}{1-l}},
$$
where $f^\leftarrow(x):=\inf\{y: f(y)>x\}$ defines the generalised inverse.
\end{lemma}
\begin{proof}
Since $|\log f| \in \RV_l$, 
Bingham's Lemma~\cite[Theorem~4.12.10]{bingham89} implies 
$\log\PP(\Vv\geq x) = \log\int_{x}^{\infty}\E^{\log f(y)}\D y \sim \log f(x)$,
as~$x$ tends to infinity,
and the result follows from Kasahara's Tauberian theorem~\cite[Theorem~4.12.7]{bingham89}.
\end{proof}

\begin{proof}[Proof of Theorem~\ref{thm: ldp_ThinTail_rescaled}]
By Lemma~\ref{lemma: log-mgf-asymptotics}, if~$\gamma+\alpha<1$, then 
$$
\Lambda^{(\alpha)}_\gamma\left(t,\frac{u}{t^\gamma}\right)
\sim
\frac{\ccp}{2\ggp}u^{2\ggp} 
t^{\gamma + [1-2(\alpha+\gamma)]\ggp},\quad\text{as }t \text{ tends to zero,
for all }u\in \RR.
$$
The only non-degenerate result is obtained when~$\alpha = \frac{1}{2}(1 - \gamma/\ggm)$, 
and the requirement that $\gamma+\alpha<1$ implies that $\gamma< \ggp$. 
The rest follows directly from the G\"artner-Ellis theorem.
If~$\gamma+\alpha = 1$, then 
$$
\Lambda^{(\alpha)}_\gamma\left(t,\frac{u}{t^\gamma}\right)
\sim \frac{\ccp}{\ggp}2^{\ggp-1}\Lambda(u)^{\ggp}
t^{\gamma-\ggp},\quad\text{as }t\text{ tends to zero, for all }u\in (u_-, u_+),
$$
which imposes $\gamma =\ggp$. 
Define now the function
$f(u) := \ggp^{-1}\ccp 2^{\ggp-1}\Lambda(u)^{\ggp}$ on~$(u_-, u_+)$;
then
$$
f'(u) = 2^{\ggp-1}\ccp\Lambda(u)^{\ggp-1}\Lambda'(u)
\qquad \text{and}\qquad 
f''(u) = 2^{\ggp-1} \ccp
\left[(\ggp-1)\Lambda'(u)^2\Lambda(u)^{\ggp-2} + \Lambda(u)^{\ggp-1}\Lambda''(u)\right].
$$
Since~$\ggp>1$, and since~$\Lambda$ is strictly convex and tends to infinity at~$u_\pm$,
then so does~$f$. 
Consequently, for any~$x\in\RR$ the equation~$x=f'(u)$ admits a unique solution in~$(u_-,u_+)$,
hence the function~$\LLp$ is well defined on~$\RR$ and is a good rate function. 
The large deviations principle follows from the G\"artner-Ellis theorem.
\end{proof}
In a mathematical finance context, the case~$\gamma<\ggm$ belongs to the so-called regime of moderately 
out-of-the-money~\cite{FrizGerholdPinter, mijatovic12},
with time-dependent log-strike~$x_t = xt^\alpha$, for~$x\in\RR_+^*$ and $\alpha\in(0,1/2)$. 
In a thin-tail randomised environment, 
the rescaled limiting cgf does not satisfy~\cite[Assumption~6.1]{FrizGerholdPinter} 
in which the limit is assumed to have a quadratic form.
Moreover, Theorem~\ref{thm: ldp_ThinTail_rescaled} implies that for the original process~$(X_t)_{t\geq 0}$,
\begin{equation}\label{eq:MOTM}
\PP\left(X_t\geq x_t\right) 
=\PP\left(X_t^{(\alpha)} \geq x\right)
= \exp\left(-\frac{\LLm(x)}{t^\gamma}(1+o(1))\right),\quad\text{as } t\text{ tends to zero.}
\end{equation}
Tail probabilities translate naturally to asymptotic behaviours of the implied volatility, 
denoted by~$\sigma_t(x)$, for given maturity~$t$ and log-strike~$x$.
The following corollary makes this statement precise:
\begin{corollary}\label{impvol:MOTM}
Consider the following two regimes:
\begin{itemize}
\item Moderately out-of-the-money (MOTM): $(\alpha,x) \in (0,1/2)\times\RR^*$;
\item Small time and large strike: $(\alpha,x) \in (1-\ggp,0)\times\RR^*$.
\end{itemize}
Under Assumption~\ref{Assu:V0}, let~$x_t:=xt^\alpha$, 
and~$\widehat{\gamma} := (1-2\alpha)(1-\ggm)>0$. 
Then 
$\displaystyle 
\lim_{t\downarrow 0} t^{\widehat{\gamma}}\sigma_t^2(x_t) = \ccm^{-1}\ggm x^{2(1-\ggm)}$.
\end{corollary}
\begin{proof}
We only prove the case $x>0$, the other cases being analogous.
For~$\gamma := \ggm(1-2\alpha)>0$,
Equation~\eqref{eq:MOTM} implies that as~$t$ tends to zero,
$-\log\PP\left(X_t\geq x_t\right)\sim t^{-\gamma}\LLm(x)$.
It is easy to check that the sequence~$(t^\gamma, x_t)_{t\geq 0}$ satisfies~\cite[Hypothesis~2.2]{caravenna15},
so that the corollary follows from~\cite[Theorem~2.3]{caravenna15}:
\begin{align*}
\sigma_t^2(x_t)
&\sim\frac{2xt^\alpha}{t}
\left(\sqrt{\frac{\ccm}{2\ggm}\frac{x^{2\ggm-1}}{t^{\gamma+\alpha}}} - \sqrt{\frac{\ccm}{2\ggm}\frac{x^{2\ggm-1}}{t^{\gamma+\alpha}}-1}\right)^2
\sim\frac{2x}{t^{1-\alpha}}
\left(\sqrt{\frac{\ccm}{2\ggm}\frac{x^{2\ggm-1}}{t^{\gamma+\alpha}}} + \sqrt{\frac{\ccm}{2\ggm}\frac{x^{2\ggm-1}}{t^{\gamma+\alpha}}-1}\right)^{-2}
&
\sim \frac{\ggm x^{2(1-\ggm)}}{\ccm t^{\widehat{\gamma}}}.
\end{align*}
\end{proof}
This result can actually be improved slightly, as follows:
\begin{corollary}\label{coro: MOTM}
Under Assumption~\ref{Assu:V0},
for any slowly varying (at zero) function~$s:\RR_+^*\to\RR_+^*$ and $\gamma \in \left(0,\ggm\right)$,
let~$\alpha:=\frac{1}{2}(1 - \gamma / \ggm)$ and $x_t:=t^\alpha s(t)$. 
Then
$$
\PP\left(X_t\geq x_t\right) = \exp\left(-\frac{\LLm\left( s(t) \right)}{t^\gamma}\left(1+o(1)\right)\right).
$$
\end{corollary}
\begin{proof}
The function $q:\RR_+^*\to\RR_+^*$ defined by $q(t):= s(t)^{-2\ggm}$
is slowly varying at zero,
and 
$\lim_{t\downarrow 0}t^\gamma q(t) 
= \lim_{t\downarrow 0} \left(t^{\frac{\gamma}{2\ggm}}/s(t)\right)^{2\ggm}
= 0.$
Notice that 
$\gamma+\alpha 
\in(1/2,1)$, so that $t=o\left(t^{\gamma+\alpha}q(t)s(t)\right)$, 
and Lemma~\ref{lemma:CDasymp} implies that the rescaled cgf of the process~$\left(X_t/(s(t)t^\alpha)\right)_{t\geq 0}$ is given by
\begin{align*}
t^\gamma q(t)\log\M\left(t,\frac{u}{t^{\gamma+\alpha}q(t)s(t)}\right)
&= t^\gamma q(t)\Crm\left(t,\frac{u}{t^{\gamma+\alpha}q(t)s(t)}\right) 
+ t^\gamma q(t)\log\M_{\Vv}\left(\Drm\left(t,\frac{u}{t^{\gamma+\alpha}q(t)s(t)}\right)\right)\\
&= \Oo\left(t^{1+2\gamma+\alpha}q(t)^2s(t)\right)
+ t^\gamma q(t)\log\M_\Vv\left(\frac{u^2 t(1+o(1))}{2 t^{2(\gamma+\alpha)} (q(t)s(t))^2}\right).
\end{align*}
Then from Lemma~\ref{lemma: log-mgf-asymptotics}, plugging in the expressions for~$\alpha$ 
and the function~$q$, the limit of the rescaled cgf reads
\begin{equation*}
\lim_{t\downarrow 0}t^\gamma q(t)\log\M \left(t,\frac{u}{t^{\alpha+\gamma}q(t)s(t)}\right) 
 = \frac{\ccp}{2\ggp} u^{2\ggp}
\lim_{t\downarrow 0} t^\gamma q(t) \left(\frac{1+o(1)}{t^{\gamma/\ggp}(q(t)s(t))^2}\right)^{\ggp}
 =  \frac{\ccp}{2\ggp} u^{2\ggp}.
\end{equation*}
The G\"artner-Ellis theorem implies that $\left(X_t/(s(t)t^\alpha)\right)\sim\LDP(t^\gamma q(t), \LLm)$,
with~$\LLm$ in~\eqref{eq:I}.
Consequently,
$$
-\inf_{x\in (1,\infty)}\LLm(x)
\leq \lim_{t\downarrow 0}t^\gamma q(t)\log\PP\left(X_t\geq x_t\right)
=\lim_{t\downarrow 0}t^\gamma q(t)\log\PP\left(\frac{X_t}{t^\alpha s(t)}\geq 1\right)
\leq-\inf_{x\in[1,\infty)}\LLm(x).
$$
The proof then follows by noticing that
$\displaystyle 
\frac{\LLm(1)}{q(t)} 
= \frac{\ccm}{2\ggm}s(t)^{2\ggm}
= \LLm\left( s(t) \right)$
for all $t>0$.
\end{proof}
%%%%%%%%%%%%%%%%%%%%%%%%%%%%%%%%%%%%%%%%%%%%%%%%%
%%%%%%%%%%%%%%%%%%%%%%%%%%%%%%%%%%%%%%%%%%%%%%%%%
\subsection{Fat-tail distribution}\label{sec:fat-tail}
The fat-tail distribution case yields some degeneracy, 
and forces us to analysis the asymptotic behaviour of the cumulant generating function in more details,
in particular using sharp large deviations techniques for the rescaled process $(\xhat_t)_{t\geq 0}$ 
defined by
$\xhat_t := g(t)^{-1} X_t$, for~$t>0$,
where the function~$g:\RR_+\to\RR_+$ satisfies
$g(t) = o(1)$ and $\sqrt{t} = o(g)$, as~$t$ tends to zero.
For any rescaling function $h(t) := \sqrt{t} / g(t)$ ($=o(1)$), 
denote the rescaled cumulant generating function as
$$
\lambhat_t(u) := h(t)\log\EE\left[\exp\left\{\frac{u}{h(t)}\xhat_t\right\}\right].
$$
We provide a full asymptotic expansion for the European call option price with a time-dependent log-strike
$x_t:=xg(t)$, for any fixed $x\neq 0$, 
and translate this into small-time asymptotic behaviour of the implied volatility~$\sigma_t(x_t)$.
We discuss the case where 
the initial randomisation satisfies Assumption~\ref{assu:fat-tail-2} with~$\omega=1$.
The case where~$\omega = 2$ can be processed in a similar fashion.
\begin{theorem}\label{thm:fat-tail-euro-call}
For any~$x\neq 0$, as~$t$ tends to zero, 
a European call option with strike~$x_t$ satisfies
\begin{equation}\label{eq:MDPasymp}
\EE\left(\E^{X_t} - \E^{x_t}\right)^+
= (1-\E^{x_t})^+ 
+ \exp\left(-\sqrt{\frac{2\m}{t}} |x_t| + \gamma_1 + x_t\right)
\frac{|x_t|^{|\gamma_0|-1}}{\Gamma(|\gamma_0|)(2\m)^{1-\gamma_0/2}}t^{1+\gamma_0/2}g(t)\left(1+o(1)\right).
\end{equation}
Moreover, the implied volatility satisfies
\begin{equation*}
\sigma_t^2(x_t) = \frac{|x_t|}{2\sqrt{2\m t}} + 
\mathrm{h}_1(x) + \mathrm{h}_2\log(t) + \frac{1}{4\m}\log(g(t)) + o(1),
\end{equation*}
where 
$$
\mathrm{h}_1(x) 
 := \frac{1}{8\m}\left\{x_t
- (2\gamma_0+1)\log |x_t|
+ \log\left(\frac{16\pi\E^{2\gamma_1}}{\Gamma(|\gamma_0|)^2}\right) 
- \left(|\gamma_0|+\frac{1}{2}\right)\log(2\m) \right\}
\quad\text{and}\quad
\mathrm{h}_2  := \frac{1}{8\m}\left(\frac{1}{2} - |\gamma_0|\right).
$$
Furthermore, under Assumption~\ref{assu:fat-tail-2},
$X^g\sim\LDP(h(t), \sqrt{2\m}|x|)$.
\end{theorem}

\begin{proof}[Proof of Theorem~\ref{thm:fat-tail-euro-call}]
The proof is close to that of~\cite[Theorem 4.10]{JS16}, so that we only sketch the highlights.
Notice that~$\sqrt{t} = o(h(t))$.
Following similar steps to~\cite[Lemma~D.1]{JS16}, it is easy to show that
for any $x\neq 0$ and small~$t>0$,
the equation~$\partial_u\lambhat_t(u) = x$ admits a unique solution~$u_t^*(x)$ satisfying
$u_t^*(x) = \sgn (x)\sqrt{2\m} -\frac{|\gamma_0|}{x}h(t) + \Oo\left(h(t)^2 + \sqrt{t}\right)$.
Then as~$t$ tends to zero, direct computations yield
\begin{equation*}
\exp\left\{\frac{-xu_t^*(x) + \lambhat_t(u_t^*(x))}{h(t)}\right\}
=\exp\left\{-\sqrt{\frac{2\m}{t}} |x_t| - \gamma_0+\gamma_1\right\}\left(\frac{|\gamma_0|\sqrt{2\m t}}{|x_t|}\right)^{\gamma_0}(1+o(1)).
\end{equation*}
For fixed~$x\neq 0$ and small $t>0$, 
define the time-dependent measure~$\QQ_t$ by
$$
\frac{\D \QQ_t}{\D \QQ}:=\exp\left\{\frac{u_t^*(x)\xhat_t - \lambhat_t(u_t^*(x))}{h(t)}\right\},
$$
so that, for~$x>0$, 
\begin{align}\label{eq:mdp_pricing}
\EE\left(\E^{X_t} - \E^{x_t}\right)^+
&= \EE^{\QQ_t}\left[\E^{x_t}\left(\E^{g(t)(\xhat_t - x)}-1\right)^+\frac{\D \QQ}{\D \QQ_t}\right]\\
&= \exp\left\{\frac{-xu_t^*(x) + \lambhat_t(u_t^*(x))}{h(t)} \right\}\E^{x_t}
\EE^{\QQ_t}\left[\exp\left(\frac{-u_t^*(x)Z_t}{h(t)}\right)\left(\E^{g(t)Z_t}-1\right)^+\right],\nonumber
\end{align}
with $Z_t:=\xhat_t - x$.
From Lemma~\ref{lemma: cf-of-z},  under the measure~$\QQ_t$,
the characteristic function of~$Z_t$ satisfies
$$
\Psi_t(u):=\EE^{\QQ_t}[\E^{\I uZ_t}] = \E^{-\I u x}\left(1-\frac{\I ux}{|\gamma_0|}\right)^{\gamma_0}\left(1+o(1)\right),\quad\text{as }t\text{ tends to zero}.
$$
By Fourier inversion, we can therefore write, for small~$t>0$,
\begin{equation*}
\EE^{\QQ_t}\left[\exp\left(\frac{-u_t^*(x)Z_t}{h(t)}\right)\left(\E^{g(t)Z_t}-1\right)^+\right]
= \frac{t}{2\pi}\int_{-\infty}^{\infty}\frac{\Psi_t(u)\D u}{(u_t^*(x) + (\I u - g(t))h(t) )(u_t^*(x) + \I uh(t))}
= \frac{t f_\Gamma(x)}{2\m}(1+o(1)),
\end{equation*}
where 
$f_{\Gamma}(y):= \frac{y^{|\gamma_0|-1}}{\Gamma(|\gamma_0|)}\exp(-|\frac{\gamma_0}{x}|y)(|\frac{\gamma_0}{x}|)^{|\gamma_0|}$, 
for $y>0$.
The result then follows directly by plugging this back into~\eqref{eq:mdp_pricing}. 
The case where~$x<0$ follows by Put-Call parity.
Finally, a direct application of~\cite[Corollary~7.2]{gao14} yields the asymptotics for the implied volatility.
\end{proof}

%%%%%%%%%%%%%%%%%%%%%%%%%%%%%%%%%%%%%%%

\appendix
\section{Useful results}\label{app:heston}
We recall the the following small-time expansion of the (rescaled) functions~$\Crm$ and~$\Drm$ 
from~\cite[Appendix C]{JS16}:
\begin{lemma}\label{lemma:CDasymp}
%Let~$\beta\in\RR$.
The following asymptotic behaviour as~$t$ tends to zero:
\item 
\begin{equation*}
\Crm\left(t,\frac{u}{h(t)}\right) 
= 
\left\{
\begin{array}{lll}
\text{undefined}, & u\neq 0, & \text{if }h(t)= o(t),\\
\displaystyle 
\Oo(1), & u\in(u_-, u_+),
& \text{if }h(t) = t+\Oo(t^2),\\
\displaystyle \Oo\left(th(t) + h^3(t)\right)
 + \frac{\kappa\theta u^2}{4}\left(\frac{t}{h(t)}\right)^{2}
 \left[1 + \Oo\left(h(t) + \frac{t}{h(t)}\right)\right],  &  u\in \RR,
 & \text{if }t=o(h(t));
\end{array}
\right.
\end{equation*}
\begin{equation*}
\Drm\left(t,\frac{u}{h(t)}\right) 
= 
\left\{
\begin{array}{lll}
0, & \text{if } u = 0, & \text{for any function }h,\\
\text{undefined}, &  u\neq 0, & \text{if } h(t)=o(t),\\
\displaystyle t^{-1}\Lambda(u) + \Oo(1), & u\in(u_-, u_+), 
& \text{if }h(t) = t + \Oo(t^2),\\%+\beta t^2+ o(t^2),\\
\displaystyle \frac{u^2 t}{2h^{2}(t)} \left[1 -\frac{h(t)}{u} + \frac{\rho\xi ut}{2h(t)}
 + \Oo\left(t + h^2(t) + \frac{t^2}{h^2(t)}\right) \right],  &u\in \RR, & \text{if } t=o(h(t)).
\end{array}
\right.
\end{equation*}
\end{lemma}

%%%%%%%%%%%%%%%%%%%%%%%%%%%%%%%%%%%%%%
We also recall the following lemma:
\begin{lemma}\label{lemma: cf-of-z}[Lemma D.3 in~\cite{JS16}]
For any~$x\neq 0$, let~$Z_t:=(X_t-x)/\vartheta(t)$, where~
$\vartheta(t):=\ind_{\{\omega=1\}} + \ind_{\{\omega= 2\}}t^{1/8}$. 
Under Assumption~\ref{assu:fat-tail-2}, as~$t$ tends to zero, 
the characteristic function of~$Z_t$ under measure~$\QQ_t$ is
\begin{equation*}
\Psi_t(u) := \EE^{\QQ_t}\left(\E^{\I u Z_t}\right)
= 
\left\{
\begin{array}{ll}
\displaystyle\E^{-\I u x}\left(1 - \frac{\I u x}{|\gamma_0|}\right)^{\gamma_0}\left(1+o(1)\right), &\text{for }\omega =1,\\
\displaystyle \exp\left(\frac{-u^2\zeta^2(x)}{2}\right)\left(1+o(1) \right),&\text{for }\omega = 2,
\end{array}
\right.
\end{equation*}
where~$\zeta(x) := \displaystyle\sqrt{2}\left(\frac{2\m}{\gamma_0^2}\right)^{1/8}|x|^{3/4}$.
\end{lemma}
%%%%%%%%%%%%%%%%%%%%%%%%%%%%%%%%%%%%%%

%%%%%%%%%%%%%%%%%%%%%%%%%%%%%%%%%%%%%%%
%%%%%%%%%%%%%%%%%%%%%%%%%%%%%%%%%%%%%%%
\end{document}